\documentclass[runningheads,11pt]{llncs}
\usepackage{amsmath}
\usepackage{graphicx}
\usepackage{fullpage}

\def\calL{\mathcal{L}}
\def\calA{\mathcal{A}}
\def\calB{\mathcal{B}}

\begin{document}

\title{Algorithms for Covering Multiple Barriers\thanks{A preliminary
version of this paper appeared in the Proceedings of the 15th Algorithms and Data Structures Symposium (WADS 2017).}
}
\author{
Shimin Li
\and
Haitao Wang
}
\institute{
Department of Computer Science\\
Utah State University, Logan, UT 84322, USA\\
\email{shiminli@aggiemail.usu.edu,haitao.wang@usu.edu}\\
}

\maketitle

\begin{abstract}
In this paper, we consider the problems for covering multiple intervals on a line.
Given a set $B$ of $m$ line segments (called ``barriers'') on a horizontal line $L$
and another set $S$ of $n$ horizontal line segments of the same length in the plane, we want to move all segments of $S$ to $L$ so that their union covers all barriers and the maximum
movement of all segments of $S$ is minimized. Previously, an $O(n^3\log n)$-time
algorithm was given for the case $m=1$. In this paper, we propose an $O(n^2\log n\log \log n+nm\log m)$-time algorithm for a more general setting with any $m\geq 1$, which also improves the previous work when $m=1$. We then consider a line-constrained version of the problem in which the segments of $S$ are all initially on the line $L$. Previously, an $O(n\log n)$-time algorithm was known for the case $m=1$. We present an algorithm of $O(m\log m+n\log m \log n)$ time for any $m\geq 1$.
These problems may have applications in mobile sensor barrier coverage in wireless sensor networks.
\end{abstract}


\section{Introduction}\label{intr}

In this paper, we study algorithms for the problems for covering multiple
barriers. These are basic geometric problems and may have applications in
barrier coverage of mobile sensors in wireless
sensor networks. For convenience, in the following we introduce and discuss the
problems from the mobile sensor barrier coverage point of view.

Let $L$ be a line, say, the $x$-axis. Let $\calB$ be a set of $m$
pairwise disjoint segments, called {\em barriers}, sorted on $L$ from left to right. Let $S$ be a set of $n$ sensors in the plane, and each sensor $s_i\in S$ is
represented by a point with coordinate $(x_i,y_i)$. If a sensor is moved on $L$, it has a
{\em sensing/covering range} of length $r$, i.e., if a sensor $s$ is
located at $x$ on $L$, then all points of $L$ in the interval
$[x-r,x+r]$ are {\em covered} by $s$ and the interval is called
the {\em covering interval} of $s$. The problem is to move all sensors of $S$ onto $L$ such that each point of every barrier is covered by at least one
sensor and the maximum movement of all sensors of $S$ is minimized, i.e., the value $\max_{s_i\in S}\sqrt{(x_i-x_i')^2+y_i^2}$ is minimized, where $x_i'$ is the location of $s_i$ on $L$ in the solution (its $y$-coordinate is $0$ since $L$ is the $x$-axis). We call it the {\em multiple-barrier coverage} problem, denoted by MBC.

We assume that covering range of the sensors is long enough so that a coverage of all barriers is always possible. Note that we can check whether a coverage is possible in $O(m+n)$ time by an easy greedy algorithm (e.g., try to cover all barriers one by one from left to right using sensors in such a way that their covering intervals do not overlap except at their endpoints).

Previously, the case $m=1$ was studied and the
problem was solved in $O(n^3\log n)$ time~\cite{ref:LiMi15}. In this
paper, we propose an $O(n^2\log n\log \log n+nm\log m)$-time algorithm for any $m\geq 1$, which
also improves the algorithm in~\cite{ref:LiMi15} by almost a linear factor
when $m=1$.

We further consider a {\em line-constrained} version of the problem
where all sensors of $S$ are initially on $L$. Previously, the case
$m=1$ was studied and the problem was solved in $O(n\log n)$ time~\cite{ref:ChenAl13}. We present an $O(m\log m+n\log m \log n)$ time\footnote{The time was erroneously claimed to be $O((m+n)\log (m+n))$ in the conference version of the paper.} algorithm for any $m\geq 1$, and the running time matches
that of the algorithm in~\cite{ref:ChenAl13} when $m=1$.

\subsection{Related Work}
\label{related}
Sensors are basic units in wireless sensor networks. The advantage of allowing
the sensors to be mobile increases monitoring capability compared to
those static ones. One of the most important applications
in mobile wireless sensor networks is to monitor a barrier to detect intruders in an
attempt to cross a specific region. Barrier coverage~\cite{ref:KumarBa05,ref:LiMi15}, which guarantees that every movement crossing a barrier of sensors will be detected, is known to be an appropriate model of coverage for such applications.
Mobile sensors normally have limited battery power and therefore their
movements should be as small as possible.

Dobrev et al.~\cite{ref:DobrevCo15} studies several problems on covering multiple barriers in the plane. They showed that these problems are generally NP-hard when  sensors have different ranges. They also proposed polynomial-time algorithms for several special cases of the problems, e.g., barriers are parallel or perpendicular to each other, and sensors have some constrained movements. In fact, if sensors have different ranges, by an easy reduction from the Partition Problem as in \cite{ref:DobrevCo15}, we can show that our problem MBC is NP-hard
even for the line-constrained version and $m=2$.

Other previous work has been focused on  the line-constrained problem with $m=1$.
Czyzowicz et al.~\cite{ref:CzyzowiczOn09} first gave an $O(n^2)$ time algorithm,
and later, Chen et al.~\cite{ref:ChenAl13} solved the problem in $O(n \log n)$ time.
If sensors have different ranges, Chen et al.  \cite{ref:ChenAl13} presented an $O(n^2\log n)$ time algorithm.
For the {\em weighted case} where sensors have weights such that the moving cost of a sensor is its moving distance times its weight, Lee et al.~\cite{ref:LeeMi17} gave an $O(n^2\log n\log\log n)$ time algorithm for the case where sensors have the same range.

The {\em min-sum} version of the line-constrained problem with $m=1$
has also been studied, where the objective is to minimize the sum of the moving distances of all sensors. If sensors have different ranges, then
the problem is NP-hard \cite{ref:CzyzowiczOn10}. Otherwise,
Czyzowicz et al. \cite{ref:CzyzowiczOn10} gave an $O(n^2)$ time
algorithm, and Andrews and Wang~\cite{ref:AndrewsMi17} improved the algorithm to $O(n\log n)$~time. The {\em min-num} version of the problem was also studied, where the goal is to move the minimum number of sensors to form a barrier coverage. Mehrandish et al. \cite{ref:MehrandishOn11,ref:MehrandishMi11} proved that the problem is NP-hard if sensors have different ranges and gave polynomial time algorithms otherwise.

Bhattacharya et al.~\cite{ref:BhattacharyaOp09} studied a circular barrier coverage problem in which the barrier is a circle and the sensors are initially located inside the circle. The goal is to move sensors to the circle to form a regular $n$-gon (so as to cover the circle) such that the maximum sensor movement is minimized. An $O(n^{3.5}\log n)$-time algorithm was given in~\cite{ref:BhattacharyaOp09} and later Chen et al.~\cite{ref:ChenOp15} improved the algorithm to $O(n\log^3 n)$ time. The min-sum version of the problem was also studied~\cite{ref:BhattacharyaOp09,ref:ChenOp15}.

\subsection{Our Approach}\label{sec:approach}

To solve the problem MBC, one major difficulty is that we do not know the
order of the sensors of $S$ on $L$ in an optimal solution. Therefore,
our main effort is to find such an order. To this end, we first
develop a {\em decision algorithm} that can determine whether $\lambda\geq
\lambda^*$ for any value $\lambda$, where $\lambda^*$ is the
maximum sensor movement in an optimal solution. Our decision algorithm runs in   $O(m+n\log n)$ time. Then, we solve the
problem MBC by ``parameterizing'' the decision
algorithm in a way similar in spirit to parametric
search~\cite{ref:MegiddoAp83}. The high-level scheme of our algorithm is very similar
to those in~\cite{ref:ChenAl13,ref:LeeMi17}, but many low-level computations are different.

The line-constrained version of the problem is
much easier due to an {\em order preserving property}:
there exists an optimal solution in
which the order of the sensors is the same as in the input. This
leads to a linear-time decision algorithm using the greedy strategy. Also based on this property, we can find a set $\Lambda$ of $O(n^2m)$
``candidate values'' such that $\Lambda$ contains $\lambda^*$. To avoid
computing $\Lambda$ explicitly, we implicitly organize the elements of $\Lambda$ into
$O(n)$ sorted arrays such that each array element can be found in $O(\log m)$ time.
Finally, by applying the matrix
search technique in \cite{ref:FredericksonGe84}, along with our linear-time decision algorithm, we compute $\lambda^*$ in $O(m\log m+n\log m \log n)$ time.
We should point out that implicitly organizing the elements of $\Lambda$ into sorted arrays
is the key and also the major difficulty for solving the problem, and our technique may be interesting in its own right.

The rest of the paper is organized as follows. We introduce some notation
in Section \ref{sec:pre}. In Section~\ref{sec:line}, we present our
algorithm for the line-constrained problem. In Section \ref{sec:decision}, we present our decision algorithm for the problem MBC. Section \ref{sec:optimization} solves the problem MBC.  We conclude the paper in Section~\ref{sec:conclusion}, with remarks that our techniques can be used to reduce the space complexities of some previous algorithms
in~\cite{ref:LeeMi17}.

\section{Preliminaries}
\label{sec:pre}

We denote the barriers of $\calB$ by $B_1,B_2,\ldots,B_m$ sorted on $L$ from left to right. For each $B_i$, let $a_i$ and $b_i$ denote the left and right endpoints of $B_i$, respectively. For ease of exposition, we make a general position assumption that $a_i\neq b_i$ for each $B_i$. The degenerated case can also be handled by our techniques, but the discussions would be more tedious.

With a little abuse of notation, for any point $x$ on $L$ (the $x$-axis), we also use $x$ to denote its $x$-coordinate, and vice versa.
We assume that the left endpoint of $B_1$ is at $0$, i.e., $a_1=0$. Let $\beta$ denote the right endpoint of $B_m$, i.e., $\beta=b_m$.

We denote the sensors of $S$ by $s_1,s_2,\ldots,s_n$ sorted by their $x$-coordinates.
For each sensor $s_i$ located on a point $x$ of $L$, $x-r$ and $x+r$ are the left and right endpoints of the covering interval of $s_i$, respectively, and we call them the {\em left and right extensions} of $s_i$, respectively.

Again, let $\lambda^*$ be the
maximum sensor movement in an optimal solution.
Given $\lambda$, the {\em decision problem} is to determine whether $\lambda\geq \lambda^*$, or equivalently, whether we can move each sensor with distance at most $\lambda$ such that all barriers can be covered. If yes, we say that $\lambda$ is a {\em feasible value}. Thus, we also call it a {\em feasibility test} on $\lambda$.

In the following, we assume that $\lambda^*>0$, i.e., we have to move at least one sensor to produce a solution. This can be easily verified in $O((m+n)\log(m+n))$ time by checking whether all barriers can be covered by all sensors in the input.

\section{The Line-Constrained Version of MBC}
\label{sec:line}

In this section, we present our algorithm for the line-constrained MBC. As in the special case $m=1$~\cite{ref:CzyzowiczOn09}, since the sensing ranges of all sensors are equal, a useful observation is that the {\em order preserving} property holds: There exists an optimal solution in which the order of the sensors is the same as in the input. Due to this property, we first give a linear-time greedy algorithm for feasibility tests.

\begin{lemma}
Given any $\lambda>0$, we can determine whether $\lambda$ is a feasible value in $O(n+m)$ time.
\end{lemma}
\begin{proof}
We first move every sensor rightwards for distance $\lambda$. Then, every sensor is allowed to move leftwards at most $2\lambda$ but is not allowed to move rightwards any more.
Next we use a greedy strategy to move sensors leftwards as little as possible to cover the currently uncovered leftmost barrier. To this end, we maintain a point $p$ on a barrier that we need to cover such that all barrier points to the left of $p$ are covered but the barrier points to the right of $p$ are not.
We consider the sensors $s_i$ and the barriers $B_j$ from left to right.

Initially, $i=j=1$ and $p=a_1$. In general, suppose $p$ is located at a barrier $B_j$ and we are currently considering $s_i$. If $p$ is at $\beta$, then we are done and $\lambda$ is feasible.  If $p$ is located at $b_j$ and $j\neq m$, then we move $p$ rightwards to $a_{j+1}$ and proceed with $j=j+1$. In the following, we assume that $p$ is not at $b_j$. Let $x_i^r=x_i+\lambda$, i.e., the location of $s_i$ after it is initially moved rightwards by $\lambda$.

\begin{enumerate}
\item
If $x_i^r+r\leq p$, then we proceed with $i=i+1$.
\item
If $x_i^r-r\leq p<x_i^r+r$, we move $p$ rightwards to $x_i^r+r$.
\item
If $x_i^r-2\lambda-r\leq p<x_i^r-r$, then we move $s_i$ leftwards such that the left extension of $s_i$ is at $p$, and we then move $p$ to the right extension of $s_i$.
\item
If $p<x_i^r-2\lambda-r$, then we stop the algorithm and report that $\lambda$ is not feasible.
\end{enumerate}

Suppose the above moved $p$ rightwards (i.e., in the second and third cases). Then, if $p\geq \beta$, we report that $\lambda$ is feasible. Otherwise, if $p$ is not on a barrier, then we move $p$ rightwards to the left endpoint of the next barrier. In either case, $p$ is now located at a barrier, denoted by $B_j$, and we increase $i$ by one. We proceed as above with $B_j$ and $s_i$.
It is easy to see that the algorithm runs in $O(n+m)$ time. \qed
\end{proof}

Let $OPT$ be an optimal solution that preserves the order of the sensors. For each $i\in [1,n]$, let $x'_i$ be the position of $s_i$ in $OPT$.
We say that a set of $k$ sensors are in {\em attached positions} if the union of their covering intervals is a single interval of length equal to $2rk$.
The following lemma is an easy extension of a similar observation for the case $m=1$ in~\cite{ref:CzyzowiczOn09}.

\begin{lemma}\label{lem:candidate}
There exists a sequence of sensors $s_i,s_{i+1},\ldots,s_j$ in attached positions in $OPT$ such that one of the following three cases holds. (a) The sensor $s_j$ is moved to the left by distance $\lambda^*$ and $x'_i=a_k+r$ for some barrier $B_k$ (i.e., the sensors from $s_i$ to $s_j$ together cover the interval $[a_k,a_k+2r(j-i+1)]$). (b) The sensor $s_i$ is moved to the right by $\lambda^*$ and $x'_j=b_k-r$  for some barrier $B_k$. (c) The sensor $s_i$ is moved rightwards by $\lambda^*$ and $s_j$ is moved leftwards by $\lambda^*$.
\end{lemma}

Cases (a) and (b) are symmetric in the above lemma.
Let $\Lambda_1$ be the set of all possible distance values introduced by $s_j$ in Case (a). Specifically, for any pair $(i,j)$ with $1\leq i\leq j\leq n$ and any barrier $B_k$ with $1\leq k\leq m$, define $\lambda(i,j,k)=x_j-(a_k+2r(j-i)+r)$. Let $\Lambda_1$ consist of $\lambda(i,j,k)$ for all such triples $(i,j,k)$. We define $\Lambda_2$ symmetrically to be the set of all possible values introduced by $s_i$ in Case (b). We define $\Lambda_3$ as the set consisting of the values $[x_j-x_i-2r(j-i)]/2$ for all pairs $(i,j)$ with $1\leq i< j\leq n$.
Clearly, $|\Lambda_3|=O(n^2)$ and both $|\Lambda_1|$ and $|\Lambda_2|$ are $O(mn^2)$.
Let $\Lambda=\Lambda_1\cup \Lambda_2\cup \Lambda_3$.

By Lemma~\ref{lem:candidate}, $\lambda^*$ is in $\Lambda$, and more specifically, $\lambda^*$ is the smallest feasible value of $\Lambda$. Hence, we can first compute $\Lambda$ and then find the smallest feasible value in $\Lambda$ by using the decision algorithm. However, that would take $\Omega(mn^2)$ time. To reduce the time, we will not compute $\Lambda$ explicitly, but implicitly organize the elements of $\Lambda$ into certain sorted arrays and then apply the matrix search technique proposed in~\cite{ref:FredericksonGe84}, which has been widely used, e.g.,~\cite{ref:FredericksonOp91,ref:FredericksonPa91}. Since we only need to deal with sorted arrays instead of more general matrices, we review the technique with respect to arrays in the following lemma.

\begin{lemma}\label{lem:msearch}{\em \cite{ref:FredericksonOp91,ref:FredericksonPa91}}
Given a set of $N$ sorted arrays of size at most $M$ each, we can compute the smallest feasible value of these arrays with $O(\log N+\log M)$ feasibility tests and the total time of the algorithm excluding the feasibility tests is $O(\tau\cdot N\cdot \log M)$, where $\tau$ is the time for evaluating each array element (i.e., given the name of an array $A$ and an index $i$ of $A$, we can compute the value of the element $A[i]$ in $O(\tau)$ time).
\end{lemma}

With Lemma~\ref{lem:msearch}, we can compute the smallest feasible values in the three sets $\Lambda_1$, $\Lambda_2$, and $\Lambda_3$, respectively, and then return the smallest one as $\lambda^*$.
For $\Lambda_3$, Chen et al.~\cite{ref:ChenAl13} (see Lemma~14) gave an approach to order in $O(n\log n)$ time the elements of $\Lambda_3$ into $O(n)$ sorted arrays of $O(n)$ elements each such that each array element can be evaluated in $O(1)$ time. Consequently, by applying Lemma~\ref{lem:msearch}, the smallest feasible value of $\Lambda_3$ can be computed in $O((n+m)\log n)$ time.

For $\Lambda_1$ and $\Lambda_2$, in the case $m=1$, the elements of each set can be easily ordered into $O(n)$ sorted arrays of $O(n)$ elements each~\cite{ref:ChenAl13}. However, in our problem setting, the problem becomes significantly more difficult if we want to obtain a subquadratic-time algorithm. Indeed, this is the main challenge of our method. In what follows, our main effort is to prove the following lemma.

\begin{lemma}\label{lem:arrays}
For the set $\Lambda_1$, in $O(m\log m)$ time, we can implicitly form a set $\calA$ of $n$ sorted arrays of at most $(n+1)\cdot m^2$ elements each such that each array element can be evaluated in $O(\log m)$ time and every element of $\Lambda_1$ is contained in one of the arrays. The same applies to the set $\Lambda_2$.
\end{lemma}

We note that our technique for Lemma~\ref{lem:arrays} might be interesting in its own right and may find other applications as well. We also remark that the set of all array elements of $\calA$ in the lemma is a superset of $\Lambda_1$ that is much larger than $\Lambda_1$. Ideally, one may want to organize the elements of the same set $\Lambda_1$ (or a relatively smaller superset) into sorted arrays. However, we do not find a good way to do so. Fortunately, our approach of using a large superset still gives a good performance in both the preprocessing time and the query time (i.e., for evaluating an array element). In fact, as will be seen later, the reason we use a large superset is to make the query time small.

Before proving Lemma~\ref{lem:arrays}, we first prove the following result.

\begin{theorem}\label{theo:line}
The line-constrained version of MBC can be solved in $O(m\log m+n\log m \log n)$ time.
\end{theorem}
\begin{proof}
It is sufficient to compute $\lambda^*$, after which we can apply the decision algorithm on $\lambda^*$ to obtain an optimal solution.

Let $\Lambda_1'$ denote the set of all elements in the arrays of $\calA$ specified in Lemma~\ref{lem:arrays}. Define $\Lambda_2'$ similarly with respect to $\Lambda_2$. By Lemma~\ref{lem:arrays}, $\Lambda_1\subseteq \Lambda_1'$ and $\Lambda_2\subseteq \Lambda_2'$. Since $\lambda^*\in \Lambda_1\cup \Lambda_2\cup \Lambda_3$, we also have  $\lambda^*\in \Lambda_1'\cup \Lambda_2'\cup \Lambda_3$. Hence, $\lambda^*$ is the smallest feasible value in $\Lambda_1'\cup \Lambda_2'\cup \Lambda_3$. Let $\lambda_1$, $\lambda_2$, and $\lambda_3$ be the smallest feasible values in the sets $\Lambda_1'$, $\Lambda_2'$, and $\Lambda_3$, respectively. As discussed before, $\lambda_3$ can be computed in $O((n+m)\log n)$ time. By Lemma~\ref{lem:arrays}, applying the algorithm in Lemma~\ref{lem:msearch} can compute both $\lambda_1$ and $\lambda_2$ in $O((n+m)(\log m+\log n)+n\log m\log nm)$ time.

We claim that $(n+m)(\log m+\log n)+n\log m\log nm=\Theta(m\log m + n\log m\log n)$. Note that $(n+m)(\log m+\log n)+n\log m\log nm=\Theta(m\log m + m\log n+ n\log m\log n+n\log^2 m)$. To prove the claim, it is sufficient to show that $m\log n+n\log^2 m=O(m\log m+n\log m\log n)$, as follows.
If $m\leq n$, then $m\log n=O(n\log m\log n)$; otherwise, $m\log n=O(m\log m)$. Therefore, we obtain $m\log n=O(m\log m+n\log m\log n)$. On the other hand, if $m\leq n^2$, then $\log m=O(\log n)$, and thus, $n\log^2 m=O(n\log m\log n)$; otherwise, it holds that $n\log^2 m=O(m)$ because both $n$ and $\log^2 m$ are bounded by $O(\sqrt{m})$. Hence, we obtain $n\log^2 m=O(m\log m+n\log m\log n)$. The claim thus follows.
\qed
\end{proof}

\subsection{Proving Lemma~\ref{lem:arrays}}

In this section, we prove Lemma~\ref{lem:arrays}. We will only prove the case for $\Lambda_1$, since the other case for $\Lambda_2$ is symmetric. Recall that $\Lambda_1=\{\lambda(i,j,k)\ |\ 1\leq i\leq j\leq n, 1\leq k\leq m\}$, where $\lambda(i,j,k)=x_j-(a_k+2r(j-i)+r)$.

For any $j$ and $k$, let $A[j,k]$ denote the list $\lambda(i,j,k)$ for $i=1,2,\ldots, j$, which is sorted in increasing order. With a little abuse of notation, let $A[j]$ denote the union of the elements in $A[j,k]$ for all $k\in [1,m]$.
Clearly, $\Lambda_1$ is the union of $A[j]$'s for all $1\leq j\leq n$.
In the following, we will organize the elements in each $A[j]$ into a sorted array $B[j]$ of size $O(nm^2)$ such that given any index $t$, the $t$-th element of $B[j]$ can be computed in $O(\log m)$ time, which will prove Lemma~\ref{lem:arrays}. Our technique replies on the following property: the difference of every two adjacent elements in each list $A[j,k]$ is the same, i.e., $2r$.

Notice that for any $k\in [1,m-1]$, the first element of $A[j,k]$ is larger than the first element of $A[j,k+1]$, and similarly, the last element of $A[j,k]$ is larger than the last element of $A[j,k+1]$.
Hence, the first element of $A[j,m]$, i.e., $\lambda(1,j,m)$, is the smallest element of $A[j]$ and the last element of $A[j,1]$, i.e., $\lambda(j,j,1)$, is the largest element of $A[j]$. Let $\lambda_{min}[j]=\lambda(1,j,m)$ and $\lambda_{max}[j]=\lambda(j,j,1)$.

For each $k\in [1,m]$, we extend the list $A[j,k]$ to a new sorted list $B[j,k]$ with the following property: (1) $A[j,k]$ is a sublist of $B[j,k]$; (2) the difference of every two adjacent elements of $B[j,k]$ is $2r$; (3) the first element of $B[j,k]$ is in $[\lambda_{min}[j],\lambda_{min}[j]+2r)$; (4) the last element of $B[j,k]$ is in $(\lambda_{max}[j]-2r,\lambda_{max}[j]+2r)$.
Specifically, $B[j,k]$ is defined as follows.
Note that $\lambda(1,j,k)$ is the first element of $A[j,k]$.
We let $\lambda(1,j,k)-\lfloor\frac{\lambda(1,j,k)-\lambda_{min}[j]}{2r}\rfloor \cdot 2r$
be the first element of $B[j,k]$. Then, the $h$-th element of $B[j,k]$ is equal to $\lambda(1,j,k)-\lfloor\frac{\lambda(1,j,k)-\lambda_{min}[j]}{2r}\rfloor \cdot 2r+ 2r\cdot (h-1)$ for any $h\in [1,\alpha[j]]$, where $\alpha[j]=1+\lfloor\frac{\lambda_{\max}[j]-\lambda_{min}[j]}{2r}\rfloor$. Hence, $B[j,k]$ has $\alpha[j]$ elements. One can verify that $B[j,k]$ has the above four properties. Note that we can implicitly create the list $B[j,k]$ in $O(1)$ time so that given any $k\in [1,m]$ and $h\in [1,\alpha[j]]$, we can obtain the $h$-th element of $B[j,k]$ in $O(1)$ time. Let $B[j]$ be the sorted list of all elements of $B[j,k]$ for all $1\leq k\leq m$. Hence, $B[j]$ has $\alpha[j]\cdot m$ elements.

Let $\sigma_j$ be the permutation of $1,2,\ldots,m$ following the sorted order of the first elements of $B[j,k]$ for all $k\in [1,m]$. For any $k\in [1,m]$, let $\sigma_j(k)$ be the $k$-th index in $\sigma_j$.
We have the following lemma.

\begin{lemma}
For any $t$ with $1\leq t \leq  \alpha[j]\cdot m$, the $t$-th smallest element of $B[j]$ is the $h_t$-th element of the list $B[j,\sigma_j(k_t)]$, where  $h_t=\lceil \frac{t}{m}\rceil$ and $k_t=t \mod m$.
\end{lemma}
\begin{proof}
Consider any $h$ with $1\leq h\leq \alpha[j]$. Denote by $B_h[j,k]$ the $h$-th element of $B[j,k]$ for each $k\in [1,m]$.  By our definition of $B[j,k]$, $B_h[j,k]\in [\lambda_{\min}[j]+2r(h-1), \lambda_{\min}[j]+2rh)$. Therefore, for any $h'<h$, it holds that $B_{h'}[j,k]<B_h[j,k']$ for any $k$ and $k'$ in $[1,m]$.
On the other hand, by the definition of $\sigma_j$, $B_h[j,\sigma(k)]<B_h[j,\sigma(k')]$ for any $1\leq k< k'\leq m$.

Based on the above discussion, one can verify that the lemma statement holds.
\qed
\end{proof}

By the preceding lemma, if the permutation $\sigma_j$ is known, we can obtain the $t$-th smallest element of $B[j]$ in $O(1)$ time for any index $t$. Computing $\sigma_j$ can be done in $O(m\log m)$ time by sorting. If we apply the sorting algorithm on every $j\in [1,n]$, then we would need $O(nm\log m)$ time. Fortunately, the following lemma implies that we only need to do the sorting once.

\begin{lemma}
The permutation $\sigma_j$ is unique for all $j\in [1,n]$.
\end{lemma}
\begin{proof}
Consider any $j_1,j_2$ in $[1,n]$ with $j_1\neq j_2$ and any $k_1,k_2$ in $[1,m]$ with $k_1\neq k_2$. For any $j$ and $k$, let $B_1[j,k]$ denote the first element of $B[j,k]$.
To prove the lemma, it is sufficient to show that $B_1[j_1,k_1]<B_1[j_1,k_2]$ if and only if $B_1[j_2,k_1]<B_1[j_2,k_2]$.

Recall that $B_1[j,k]=\lambda(1,j,k)-\lfloor\frac{\lambda(1,j,k)-\lambda_{min}[j]}{2r}\rfloor \cdot 2r$ and $\lambda(1,j,k)=x_j-(a_k+2rj-r)$. Thus, $B_1[j,k]=x_j-a_k+r-\lfloor\frac{x_j-a_k+r-\lambda_{min}[j]}{2r}\rfloor \cdot 2r$. Further, since $\lambda_{min}[j]=\lambda(1,j,m)=x_j-(a_m+2rj-r)$, $B_1[j,k]=x_j-a_k+r-\lfloor\frac{a_m-a_k+2rj}{2r}\rfloor \cdot 2r=x_j-a_k+r-\lfloor\frac{a_m-a_k}{2r}\rfloor \cdot 2r-2rj$.

Therefore, $B_1[j_1,k_1]-B_1[j_1,k_2]=a_{k_2}-a_{k_1}+(\lfloor\frac{a_m-a_{k_2}}{2r}\rfloor - \lfloor\frac{a_m-a_{k_1}}{2r}\rfloor )\cdot 2r$ and $B_1[j_2,k_1]-B_1[j_2,k_2]=a_{k_2}-a_{k_1}+(\lfloor\frac{a_m-a_{k_2}}{2r}\rfloor - \lfloor\frac{a_m-a_{k_1}}{2r}\rfloor )\cdot 2r$.
Hence, $B_1[j_1,k_1]-B_1[j_1,k_2]=B_1[j_2,k_1]-B_1[j_2,k_2]$, which implies that  $B_1[j_1,k_1]<B_1[j_1,k_2]$ if and only if $B_1[j_2,k_1]<B_1[j_2,k_2]$.
\qed
\end{proof}

In summary, after $O(m\log m)$ time preprocessing to compute the permutation $\sigma_j$ for any $j$, we can form the arrays $B[j]$ for all $j\in [1,n]$ such that given any $j\in [1,n]$ and $t\in [1,\alpha[j]\cdot  m]$, we can compute the $t$-th smallest element of $B[j]$ in $O(1)$ time.
However, we are not done yet, because we do not have a reasonable upper bound for $\alpha[j]$, which is equal to $1+\lfloor\frac{\lambda_{\max}[j]-\lambda_{min}[j]}{2r}\rfloor
=1+\lfloor\frac{\lambda(j,j,1)-\lambda(1,j,m)}{2r}\rfloor
=j+\lfloor\frac{a_m-a_1}{2r}\rfloor$.
To address the issue, in the sequel, we will partition the indices $k\in [1,m]$ into groups and then apply our above approach to each group so that the $\alpha[j]$ values can be bounded, e.g., by $O(mn)$.

\paragraph{The Group Partition Technique.}
We consider any index $j\in [1,m]$.

We partition the indices $1,2,\ldots,m$ into groups each consisting of a sequence of consecutive indices, such that each group has the following {\em intra-group overlapping property}: For any index $k$ that is not the largest index in the group, the first element of $A[j,k]$ is smaller than or equal to the last element of $A[j,k+1]$ plus $2r$, i.e., $\lambda(1,j,k)\leq \lambda(j,j,k+1)+2r$. Further, the groups have the following {\em inter-group non-overlapping property}: For the largest index $k$ in a group that is not the last group, the first element of $A[j,k]$ is larger than the last element of $A[j,k+1]$ plus $2r$, i.e., $\lambda(1,j,k)> \lambda(j,j,k+1)+2r$.

We compute the groups in $O(m)$ time as follows.
Initially, add $1$ into the first group $G_1$. Let $k=1$. While the first element of $A[j,k]$ is smaller than or equal to the last element of $A[j,k+1]$ plus $2r$, we add $k+1$ into $G_1$ and reset $k=k+1$. After the while loop, $G_1$ is computed. Then, starting from $k+1$, we compute $G_2$ and so on until index $m$ is included in the last group. Let $G_1,G_2,\ldots,G_l$ be the $l$ groups we have computed. Note that $l\leq m$.

Consider any group $G_g$  with $1\leq g\leq l$. We process the lists $A[j,k]$ for all  $k\in G_g$ in the same way as discussed before. Specifically, for each $k\in G_g$, we create a new list $B[j,k]$ from $A[j,k]$. Based on the new lists in the group $G_g$, we form the sorted array $B_g[j]$ with a total of $|G_g|\cdot \alpha_g[j]$ elements, where $|G_g|$ is the number of indices of $G_g$ and $\alpha_g[j]$ is corresponding $\alpha[j]$ value as defined before but only on the group $G_g$, i.e., if $k_1$ and $k_2$ are the smallest and largest indices of $G_g$ respectively, then $\alpha_g[j]=1+\lfloor\frac{\lambda(j,j,k_1)-\lambda(1,j,k_2)}{2r}\rfloor$.
Let $B[j]$ be the sorted list of all elements in the lists $B_g[j]$ for all groups. Due to the intra-group overlapping property of each group, it holds that $\alpha_g[j]\leq |G_g|\cdot (n+1)$. Thus, the size of $B[j]$ is at most $\sum_{g=1}^l|G_g|^2 \cdot (n+1)$, which is at most $m^2 (n+1)$ since $\sum_{g=1}^l|G_g|=m$.

Suppose we want to find the $t$-th smallest element of $B[j]$. As preprocessing, we compute a sequence of values $\beta_g[j]$ for $g=1,2,\ldots,l$, where $\beta_g[j]=\sum_{g'=1}^g\alpha_{g'}[j]\cdot |G_{g'}|$, in $O(m)$ time. To compute the $t$-th smallest element of $B[j]$, we first do binary search on the sequence $\beta_1[j],\beta_2[j],\ldots,\beta_l[j]$ to find in $O(\log l)$ time the index $g$ such that $t\in (\beta_{g-1}[j],\beta_g[j]]$. Due to the inter-group non-overlapping property of the groups, the $t$-th smallest element of $B[j]$ is the $(t-\beta_{g-1}[j])$-th element in the array $B_g[j]$, which can be found in $O(1)$ time. As $l\leq m$, the total time for computing the $t$-th smallest element of $B[j]$ is $O(\log m)$.

The above discussion is on any single index $j\in [1,n]$.
With $O(m\log m)$ time preprocessing, given any $t$, we can find the $t$-th smallest value of $B[j]$ in $O(\log m)$ time.

For all indices $j\in [1,n]$, it appears that we have to do the group partition for every $j\in [1,n]$, which would take quadratic time. To resolve the problem, we show  that it is sufficient to only use the group partition based on $j=n$ for all other $j\in [1,n-1]$. The details are given below.

Suppose from now on $G_1,G_2,\ldots,G_l$ are the groups computed as above with respect to $j=n$. We know that the inter-group non-overlapping property holds with respect to the index $n$. The following lemma shows that the property also holds with respect to any other index $j\in [1,n-1]$.

\begin{lemma}\label{lem:nonoverlap}
The inter-group non-overlapping property holds for any $j\in [1,n-1]$.
\end{lemma}
\begin{proof}
Consider any $j\in [1,n-1]$ and any $k$ that is the largest index in a group $G_g$ with $g\in [1,l-1]$. The goal is to show that the first element of $A[j,k]$ is larger than the last element of $A[j,k+1]$ plus $2r$, i.e., $\lambda(1,j,k)> \lambda(j,j,k+1)+2r$.
Since the groups are defined with respect to the index $n$, it holds that $\lambda(1,n,k)> \lambda(n,n,k+1)+2r$.

Recall that $\lambda(i,j,k)=x_j-(a_k+2r(j-i)+r)$. Therefore, $\lambda(1,j,k)- \lambda(j,j,k+1)=a_{k+1}-a_k+2r(1-j)$ and $\lambda(1,n,k)-\lambda(n,n,k+1)=a_{k+1}-a_k+2r(1-n)$. Since $\lambda(1,n,k)> \lambda(n,n,k+1)+2r$, $a_{k+1}-a_k+2r(1-n)>2r$. As $j<n$,  $a_{k+1}-a_k+2r(1-j)>2r$, and thus $\lambda(1,j,k)> \lambda(j,j,k+1)+2r$.
\qed
\end{proof}

Consider any group $G_g$ with $1\leq g\leq l$ and any $j\in [1,n]$. For each $k\in G_g$, we create a new list $B[j,k]$ based on $A[j,k]$ in the same way as discussed before. Based on the new lists, we form the sorted array $B_g[j]$ of $|G_g|\cdot \alpha_g[j]$ elements. We also define the value $\beta_g[j]$ in the same way as before. The following lemma shows that $\alpha_g[j]$ and $\beta_g[j]$ can be computed based on $\alpha_g[n]$ and $\beta_g[n]$.

\begin{lemma}\label{lem:80}
For any $j\in [1,n-1]$ and $g\in[1,l]$, $\alpha_g[j]=\alpha_g[n]-n+j$ and $\beta_g[j]=\beta_g[n]+\delta_g\cdot (j-n)$, where $\delta_g=\sum_{g'=1}^g|G_{g'}|$.
\end{lemma}
\begin{proof}
Consider any $g\in [1,l]$. Let $k_1$ and $k_2$ be the smallest and largest indices in $G_g$, respectively. By definition, $\alpha_g[j]=1+\lfloor\frac{\lambda(j,j,k_1)-\lambda(1,j,k_2)}{2r}\rfloor
=1+\lfloor\frac{a_{k_2}-a_{k_1}+2r(j-1)}{2r}\rfloor
=j+\lfloor\frac{a_{k_2}-a_{k_1}}{2r}\rfloor$.
Therefore, for any $j\in [1,n-1]$, $\alpha_g[j]=\alpha_g[n]-n+j$.

By definition, $\beta_g[j]=\alpha_1[j]\cdot |G_1|+\alpha_2[j]\cdot |G_2|+\cdots+\alpha_g[j]\cdot |G_g|
=(\alpha_1[n]-n+j)\cdot |G_1|+(\alpha_2[n]-n+j)\cdot |G_2|+\cdots+(\alpha_g[n]-n+j)\cdot |G_g|=\beta_g[n]+(j-n)\cdot (|G_1|+|G_2|+\cdots+|G_g|)=\beta_g[n]+\delta_g\cdot (j-n)$.
\qed
\end{proof}

For each group $G_g$, we compute the permutation for the lists $B[n,k]$ for all $k$ in the group. Computing the permutations for all groups takes $O(m\log m)$ time.
Also as preprocessing, we first compute $\delta_g$, $\alpha_g(n)$ and $\beta_g(n)$ for all $g\in [1,l]$ in $O(m)$ time. By Lemma~\ref{lem:80}, for any $j\in [1,n]$ and any $g\in [1,l]$, we can compute $\alpha_g[j]$ and $\beta_g[j]$ in $O(1)$ time. Because the lists $B[n,k]$ for all $k$ in each group $G_g$ have the intra-group overlapping property, it holds that $\alpha_g[n]\leq |G_g|\cdot (n+1)$. Hence, $\sum_{g=1}^l\alpha_g[n]\leq m(n+1)$.  For any $j\in [1,n-1]$, by Lemma~\ref{lem:80}, $\alpha_g[j]<\alpha_g[n]$, and thus $\sum_{g=1}^l\alpha_g[j]\leq m(n+1)$. Recall that $B[j]$ is the sorted array of all elements in $B_g[j]$ for $g\in [1,l]$. Thus, $B[j]$ has at most $m^2(n+1)$ elements.

For any $j\in [1,n]$ and any $t\in [1,\sum_{g=1}^l|G_g|\cdot \alpha_g[j]]$, suppose we want to compute the $t$-th smallest element of $B[j]$. Due to the inter-group non-overlapping property in Lemma~\ref{lem:nonoverlap}, we can still use the previous binary search approach. For the running time, since we can obtain each $\beta_g[j]$ for any $g\in [1,l]$ in $O(1)$ time by Lemma~\ref{lem:80}, we can still compute the $t$-th smallest element of $B[j]$ in $O(\log m)$ time.

This completes the proof of Lemma~\ref{lem:arrays}.

\section{The Decision Problem of MBC}
\label{sec:decision}

In this section, we present an $O(m+n\log n)$-time algorithm for the decision problem of MBC: given any value $\lambda>0$, determine whether $\lambda\geq \lambda^*$.
Our algorithm for MBC in Section \ref{sec:optimization} will make use of this decision algorithm. The decision problem may have independent interest because in some applications each sensor has a limited energy $\lambda$ and we want to know whether their energy is enough for them to move to cover all barriers.

Consider any value $\lambda>0$. We assume that $\lambda\geq
\max_{1\leq i\leq n}|y_i|$ since otherwise some sensor cannot reach
$L$ by moving $\lambda$ (and thus $\lambda$ is not feasible).
For any sensor $s_i\in S$,
define $x_i^r=x_i+\sqrt{\lambda^2-y_i^2}$ and $x_i^l=x_i-\sqrt{\lambda^2-y_i^2}$.  Note that $x_i^r$ and $x_i^l$ are respectively the rightmost and leftmost points of $L$ $s_i$ can reach with respect to $\lambda$. We call $x_i^r$ the {\em rightmost (resp., leftmost) $\lambda$-reachable location} of $s_i$ on $L$. For any point $x$ on $L$, we use
$p^+(x)$ to denote a point $x'$ such that $x'>x$ and $x'$ is infinitesimally close to $x$.

The high-level scheme of our algorithm is similar to that in
\cite{ref:LeeMi17}. A major difference is that we need to handle the gaps between barriers.
In the following, while we will emphasize the difference between our algorithm with that in~\cite{ref:LeeMi17}, to make the paper more self-contained, we will also include certain details that may overlap with those in~\cite{ref:LeeMi17}.
We first describe the algorithm and then show its correctness. The implementation of the algorithm will be discussed afterwards.

\subsection{The Algorithm Description}
\label{sec:description}

We use a {\em configuration} to refer to a specification on where each sensor
$s_i\in S$ is located. For example, in the {\em input configuration}, each $s_i$
is at $(x_i,y_i)$.

We begin with moving each sensor $s_i$ to $x_i^r$ on $L$.  Let $C_0$ denote the
resulting configuration. In $C_0$, each sensor $s_i$ is not
allowed to move rightwards but can move leftwards on $L$ by a maximum distance $2\sqrt{\lambda^2-y_i^2}$.

If $\lambda\geq \lambda^*$, our algorithm will compute a subset of sensors with
their new locations to cover all barriers of $\calB$ and the maximum movement of each sensor in the subset is at most $\lambda$.

For each step $i$ with $i\geq 1$, let $C_{i-1}$ be the configuration
right before the $i$-th step. Our algorithm maintains the following {\em
invariants}. (1) We have a subset of sensors
$S_{i-1}=\{s_{g(1)},s_{g(2)},\ldots, s_{g(i-1)}\}$, where for each $1\leq j\leq
i-1$, $g(j)$ is the index of the sensor $s_{g(j)}$ in $S$.  (2) In $C_{i-1}$, each sensor $s_k$ of $S_{i-1}$ is at a new location $x'_k\in [x^l_k,x^r_k]$, and all other sensors are still in their locations of $C_0$.
(3) A value $R_{i-1}$ is maintained such that $0\leq R_{i-1}<\beta$, $R_{i-1}$ is on a barrier, and every barrier point $x< R_{i-1}$ is covered by a sensor of $S_{i-1}$ in $C_{i-1}$. (4) If $R_{i-1}$ is not at the left endpoint of a barrier, then
$R_{i-1}$ is covered by a sensor of $S_{i-1}$ in $C_{i-1}$. (5) The point $p^+(R_{i-1})$ is not covered by any sensor in $S_{i-1}$.

Initially when $i=1$, we let $S_0=\emptyset$ and $R_0=0$, and thus
all algorithm invariants hold for $C_0$. The $i$-th step of the algorithm finds a
sensor $s_{g(i)}\in S\setminus S_{i-1}$ and moves it to a new location $x'_{g(i)}\in
[x^l_{g(i)},x^r_{g(i)}]$ and thus obtains a new configuration $C_i$.
The details are given below.

Define $S_{i1}$ to be the set of sensors that cover the point
$p^+(R_{i-1})$ in $C_{i-1}$, i.e.,
$S_{i1}=\{s_k\ |\ x_k^r-r\leq R_{i-1} < x_k^r+r\}$. By the
algorithm invariant (5), no sensor in $S_{i-1}$ covers $p^+(R_{i-1})$.
Thus, $S_{i1}\subseteq S\setminus S_{i-1}$.
If $S_{i1}\neq \emptyset$, then we choose an {\em arbitrary} sensor
in $S_{i1}$ as $s_{g(i)}$
(e.g., see Fig.~\ref{fig:decisionalgo1}) and let $x'_{g(i)}=x^r_{g(i)}$.
We then set $R_i=x'_{g(i)}+r$, i.e., $R_i$ is at the right endpoint of the covering interval of $s_{g(i)}$. Note that $C_i$ is identical to $C_{i-1}$ because $s_{g(i)}$ is not moved.

If $S_{i1}=\emptyset$, then we define $S_{i2}=\{s_k\ |\ x^l_k-r\leq
R_{i-1}<x_k^r-r\}$ (i.e., $S_{i2}$ consists of those sensors $s_k$ that does not
cover $R_{i-1}$ when it is at $x_k^r$ but is possible to do so when it is at some location in $[x_k^l,x_k^r]$).
If $S_{i2}\neq \emptyset$, we choose the {\em leftmost}
sensor of $S_{i2}$ as $s_{g(i)}$ (i.e., sensor whose location in $C_{i-1}$ has the smallest $x$-coordinate; e.g., see Fig.~\ref{fig:decisionalgo2}), and let
$x'_{g(i)}=R_{i-1}+r$ (i.e., we move $s_{g(i)}$ to $x'_{g(i)}$ and thus obtain $C_i$).
If $S_{i2}= \emptyset$, then we conclude that $\lambda<\lambda^*$ and
terminate the algorithm.

Hence, if $S_{i1}=S_{i2}=\emptyset$, the algorithm will stop and report
$\lambda<\lambda^*$. Otherwise, a sensor $s_{g(i)}$ is found from either $S_{i1}$ or $S_{i2}$, and it is moved to $x'_{g(i)}$.
In either case, $R_i=x'_{g(i)}+r$ and $S_i=S_{i-1}\cup \{s_{g(i)}\}$. If $R_i\geq \beta$, then we terminate the algorithm and report $\lambda\geq \lambda^*$.
Otherwise, we further perform the following {\em jump-over procedure}:
We check whether $R_i$ is located at the interior of any barrier; if not,
then we set $R_i$ to the left endpoint of the barrier right after $R_i$.

This finishes the $i$-th step of our algorithm. One can verify that all algorithm invariants are maintained. As there are $n$ sensors in $S$, the algorithm will finish in at most
$n$ steps.

\begin{figure}[t]
\begin{minipage}[t]{0.39\linewidth}
\begin{center}
\includegraphics[totalheight=0.8in]{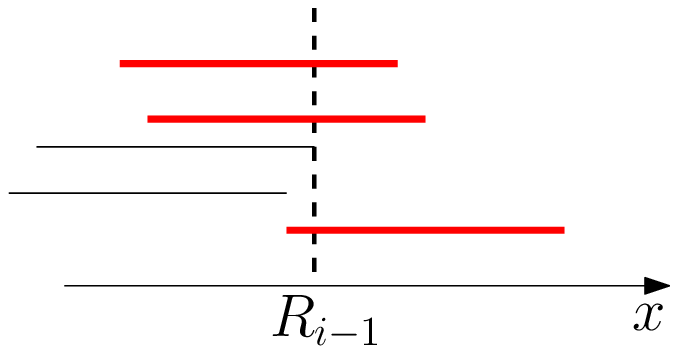}
\caption{\footnotesize Illustrating the set $S_{i1}$.
The covering intervals of sensors are shown with segments (the red thick segments
correspond to the sensors in $S_{i1}$).
Every sensor in $S_{i1}$ can be $s_{g(i)}$. }
\label{fig:decisionalgo1}
\end{center}
\end{minipage}
\hspace{0.05in}
\begin{minipage}[t]{0.59\linewidth}
\begin{center}
\includegraphics[totalheight=0.8in]{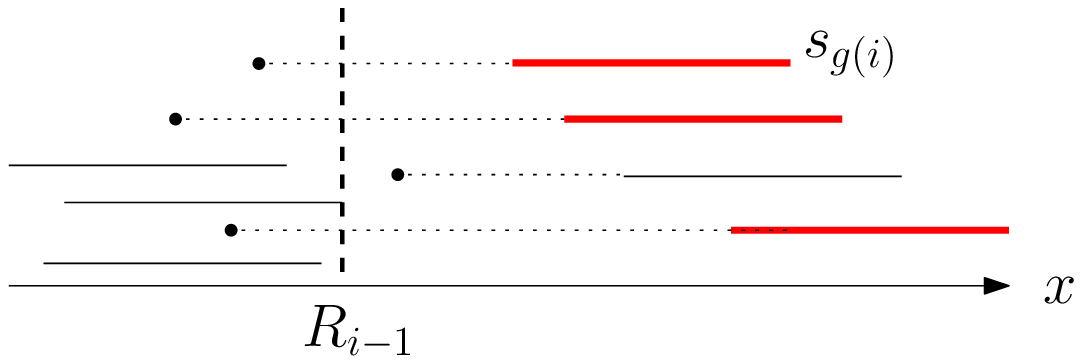}
\caption{\footnotesize Illustrating the set $S_{i2}$.
The segments are the covering intervals of sensors. The red thick segments
correspond to the sensors in $S_{i2}$. The four black points corresponding to the values $x_k^l-r$ of the four sensors $x_k$ to the right of $R_{i-1}$.
The sensor $s_{g(i)}$ is labeled. }
\label{fig:decisionalgo2}
\end{center}
\end{minipage}
\vspace*{-0.15in}
\end{figure}

\subsection{The Algorithm Correctness}
\label{sec:correct}
The correctness proof is similar to that for the algorithm
in~\cite{ref:LeeMi17}, so we briefly discuss it.

If the decision algorithm reports $\lambda\geq \lambda^*$, say, in the $i$-th step, then according to our algorithm, the configuration $C_i$ is a feasible solution. Below, we show that if the algorithm reports $\lambda<\lambda^*$, then $\lambda$ is indeed not a feasible value.

We first note that due to our jump-over procedure and our general position assumption, $R_i$ cannot be at the right endpoint of a barrier, and thus $p^+(R_i)$ must be a point of a barrier.

An interval on $L$ is said to be {\em left-aligned} if its left side is closed and equal to $0$ and its right side is open. The algorithm correctness will be easily shown with the following Lemma~\ref{lem:10}. The proof of the lemma is very similar to Lemma 1 in~\cite{ref:LeeMi17}, so we omit it.

\begin{lemma}\label{lem:10}
Consider any configuration $C_i$. Suppose $S'_i$ is the set of sensors in $S$ whose
right extensions are at most $R_i$ in $C_i$. Then, the interval $[0,R_i)$ is
the largest possible left-aligned interval such that all barrier points in the interval can be covered by the sensors of $S'_i$ with respect to $\lambda$ (i.e., the moving distance of each sensor of $S'_i$ is at most $\lambda$).
\end{lemma}

Suppose our algorithm reports $\lambda<\lambda^*$ in the $i$-th step. We show that $\lambda$ is not a feasible value. Indeed,
according to our algorithm, $R_{i-1}<\beta$ and $S_{i1}=S_{i2}=\emptyset$ in the configuration $C_{i-1}$.
Let $S_{i-1}'$ be the set of sensors whose right extensions are at most
$R_{i-1}$ in $C_{i-1}$. On the one hand, by Lemma~\ref{lem:10} (replacing index $i$ in the lemma by $i-1$), $[0,R_{i-1})$ is the largest left-aligned
interval such that all barrier points in the interval that can be covered by the sensors in $S'_{i-1}$. On the other hand, since
both $S_{i1}$ and $S_{i2}$ are empty, no sensor in $S\setminus S'_{i-1}$ can cover the point $p^+(R_{i-1})$. Recall that $p^+(R_{i-1})$ is a barrier point not covered by any sensor in $S_{i-1}$.
Due to $R_{i-1}<\beta$, we conclude that sensors of $S$ cannot
cover all barrier points in the interval $[0,p^+(R_{i-1})]\subseteq[0,\beta]$ with respect to $\lambda$. Thus, $\lambda$ is not a feasible value.
This establishes the correctness of our decision algorithm.

\subsection{The Algorithm Implementation}
\label{sec:implement}

The implementation is similar to that in~\cite{ref:LeeMi17} and we briefly discuss it. We first implement the algorithm in $O(m+n\log n)$ time, and then we reduce the time to $O(m+n\log\log n)$ under certain assumption. Then latter result will be useful in Section~\ref{sec:optimization}.

We first move each sensor $s_i$ to $x_i^r$ and thus obtain the configuration $C_0$. Then, we sort the extensions of all sensors in $C_0$, and merge the sorted list with the sorted sequence of the endpoints of all barriers. To maintain the set $S_{i1}$ during the algorithm, we sweep a point $p$ on $L$ from left to right. During the sweeping, when $p$ encounters the left (resp., right) extension of a sensor, we insert the sensor into $S_{i1}$ (resp., delete it from $S_{i1}$). In this way, in each $i$-th step of the algorithm, when $p$ is at $R_{i-1}$, $S_{i1}$ is available.

If $S_{i1}\neq \emptyset$, we  pick an arbitrary sensor in
$S_{i1}$ as $s_{g(i)}$. To store the set $S_{i1}$, since sensors
have the same range, the earlier a sensor is inserted into $S_{i1}$, the earlier it is deleted from $S_{i1}$. Thus, we can simply use a
first-in-first-out queue to store $S_{i1}$ such that each insertion/deletion can be done in constant time. We can always take the front sensor in the queue as $s_{g(i)}$.

If $S_{i1}= \emptyset$,
then we need to compute $S_{i2}$.
To maintain $S_{i2}$ during the sweeping of $p$, we do the following.
Initially when we do the sorting as discussed above, we also sort the $n$ values $x_i^l-r$ for all $1\leq i\leq n$. During
the sweeping of $p$, if $p$ encounters a point $x_k^l-r$ for
some sensor $s_k$, we insert $s_k$ to $S_{i2}$, and if $p$ encounters
a left extension of some sensor $s_k$, we delete $s_k$ from $S_{i2}$.
In this way, when $p$ is at $R_{i-1}$, $S_{i2}$ is available. If $S_{i2}\neq \emptyset$, we need to
find the leftmost sensor in $S_{i2}$ as $s_{g(i)}$, for which we use a balanced binary search tree $T$
to store all sensors of $S_{i2}$ where the ``key'' of each sensor $s_k$ is the
value $x_k^r$. $T$ can support each of the following operations on $S_{i2}$ in $O(\log n)$ time: inserting a sensor, deleting a sensor, finding the leftmost sensor.

If $s_{g(i)}$ is from $S_{i1}$, then we do not need to
move $s_{g(i)}$. We proceed to sweep $p$ as usual. If $s_{g(i)}$ is
from $S_{i2}$, we need to move $s_{g(i)}$ leftwards to
$x'_{g(i)}=R_{i-1}+r$. Since $s_{g(i)}$ is moved, we should also update the
original sorted list including the extensions of all sensors in
$C_0$ to guide the future sweeping of $p$. To avoid the explicit update, we
use a flag table for all sensor extensions in $C_0$. Initially,
every table entry is {\em valid}. If $s_{g(i)}$ is moved, then we set the
table entries of the two extensions of the sensor {\em invalid}. During the sweeping of $p$, when $p$ encounters a sensor extension, we first check the table to see whether the extension is still valid. If yes, then we proceed as usual; otherwise we ignore the event. This only costs extra constant time for each event.
In addition, we calculate $R_i$ as discussed before, and the jump-over procedure can be implemented in $O(1)$ time since the barrier endpoints are also sorted.

To analyze the running time, since the barriers are given sorted on $L$, the sorting step takes $O(m+ n\log n)$. Since there are $O(n)$ operations on the tree $T$, the total time of the algorithm is $O(m+ n\log n)$.
Thus we obtain the following result.

\begin{theorem}\label{theo:10}
Given any value $\lambda$, we can determine whether $\lambda\geq
\lambda^*$ in  $O(m+n\log n)$ time.
\end{theorem}

Our algorithm in Section \ref{sec:optimization} will perform feasibility tests multiple times, for which we have the following result.

\begin{lemma}\label{lem:decalgo2}
Suppose the values $x_i^r$ for all $i=1,2,\ldots,n$ are already sorted.
We can determine whether $\lambda\geq
\lambda^*$ in  $O(m+n\log\log n)$ time for any $\lambda$.
\end{lemma}
\begin{proof}
Our $O(m+n\log n)$ time implementation is dominated by two parts. The first
part is the sorting. The second part is on performing the operations on the set $S_{i2}$, each taking $O(\log n)$ time by using the tree $T$. The
rest of the algorithm together takes $O(n+m)$ time. Because the values $x_i^r$ for all $i=1,2,\ldots,n$ are already sorted and the barriers are also given sorted, the sorting step takes $O(n+m)$ time.

Recall that the keys of the sensors of $T$ are the values
$x_k^r$. Let $Q =\{x_k^r\ |\ 1\leq k\leq n \}$. For each sensor $s_k$, we use $rank(s_k)$ to denote the {\em rank} of $x_k^r$ in $Q$ (i.e., $rank(s_k)=t$ if $x_k^r$ is the $t$-th smallest value in $Q$). Since $Q$ is already sorted, all sensor ranks can be computed in $O(n)$ time.  It is
easy to see that the leftmost sensor of $T$ is the sensor with the
smallest rank. Therefore, we can also use the ranks as the keys of sensors of $T$, and the advantage of doing so is that the rank
of each sensor is an integer in $[1,n]$. Hence, instead of using a
balanced binary search tree, we can use an integer data structure,
e.g., the van Emde Boas Tree (or vEB tree for short) \cite{ref:CLRS09},
to maintain $S_{i2}$. The vEB tree can support each of the following operations
on $S_{i2}$ in $O(\log\log n)$ time \cite{ref:CLRS09}:
inserting a sensor, deleting a sensor, and
finding the sensor with the smallest rank. Using a vEB tree,
all operations on $S_{i2}$ in the algorithm can be
performed in $O(n\log\log n)$ time.
The lemma thus follows.
\qed
\end{proof}

\section{Solving the Problem MBC}
\label{sec:optimization}
In this section, we solve the problem MBC. It suffices to compute $\lambda^*$.

It might be tempting to use the similar algorithm scheme as that for the line-constrained version of the problem in Section~\ref{sec:line}, i.e., first identify a set of candidate values for $\lambda^*$ and then search the set to find $\lambda^*$ by using the decision algorithm. However, we find some obstacles to do so. For example, as mentioned in Section~\ref{sec:approach}, one major difficulty of the problem is that we do not know the order of the sensors of $S$ on $L$ in an optimal solution. Because of that, it is not clear how to identify a relatively small set (e.g., of polynomial size) of candidate values for $\lambda^*$. Therefore, we have to resort to other techniques to tackle the problem. The high-level scheme of our algorithm is similar to that in~\cite{ref:LeeMi17}, although some low-level details are different.

In this section, we use $x_i^r(\lambda)$ to refer to $x_i^r$ for any $\lambda$, so that we consider $x_i^r(\lambda)$ as a function on $\lambda \in[0,\infty]$, which actually defines a half of the upper branch (on the right side of the $y$-axis) of a hyperbola (in contrast, the corresponding function in~\cite{ref:LeeMi17} is linear, and thus our result is more general than that in~\cite{ref:LeeMi17}). Let $\sigma$ be the order of the values $x_i^r(\lambda^*)$ for all $i\in [1,n]$. To make use of Lemma~\ref{lem:decalgo2}, we first run a preprocessing step in Lemma~\ref{lem:pre}. Note that Lemma~\ref{lem:pre} is not from~\cite{ref:LeeMi17}, and we will discuss in Section~\ref{sec:conclusion} that the technique can be used to reduce the space complexity of the algorithm in~\cite{ref:LeeMi17}.

\begin{lemma}\label{lem:pre}
With $O(n\log^3 n+m\log^2 n)$ time preprocessing, we can compute $\sigma$ and an interval $(\lambda^*_1,\lambda^*_2]$ containing $\lambda^*$ such that $\sigma$ is also the order of the values $x_i^r(\lambda)$ for any $\lambda\in (\lambda^*_1,\lambda^*_2]$.
\end{lemma}
\begin{proof}
To compute $\sigma$, we apply Megiddo's parametric search~\cite{ref:MegiddoAp83} to sort the values $x_i^r(\lambda^*)$ for $i\in [1,n]$, using the decision algorithm in Theorem~\ref{theo:10}.
Indeed, recall that $x_i^r(\lambda)=x_i+ \sqrt{\lambda^2-y_i^2}$. Hence, as $\lambda$ increases, $x_i^r(\lambda)$ is a (strictly) increasing function.
For any two indices $i$ and $j$, there is at most one root on $\lambda\in [0,\infty)$ for the equation: $x_i^r(\lambda)=x_j^r(\lambda)$. Therefore, we can apply Megiddo's parametric search~\cite{ref:MegiddoAp83} to do the sorting. The total time is $O((\tau+n)\log^2n)$, where $\tau$ is the running time of the decision algorithm. By Theorem~\ref{theo:10}, $\tau=O(m+n\log n)$. Hence, the total time for computing $\sigma$ is $O(m\log^2 n+n\log^3 n)$.

In addition, Megiddo's parametric search~\cite{ref:MegiddoAp83} will return an interval $(\lambda^*_1,\lambda^*_2]$ such that it contains $\lambda^*$ and $\sigma$ is also the order of the values $x_i^r(\lambda)$ for any $\lambda\in (\lambda^*_1,\lambda^*_2]$.
\qed
\end{proof}

Note that $\lambda^*$ is the smallest feasible value. As $\lambda^*\in (\lambda^*_1,\lambda^*_2]$, our subsequent feasible tests will be only on values $\lambda\in (\lambda^*_1,\lambda^*_2)$ because if $\lambda\leq \lambda^*_1$, then $\lambda$ is not feasible and if $\lambda\geq \lambda^*_2$, then $\lambda$ is feasible.
Lemmas \ref{lem:decalgo2} and \ref{lem:pre} together lead to the following result.

\begin{lemma}\label{lem:faster}
Each feasibility test can be done in $O(m+n\log\log n)$ time for any $\lambda\in (\lambda^*_1,\lambda^*_2)$.
\end{lemma}

To compute $\lambda^*$, we ``parameterize'' our decision algorithm with $\lambda$ as a parameter. Although we do not know $\lambda^*$, we execute the decision algorithm in such a way that it computes the same subset of sensors $s_{g(1)},s_{g(2)},\ldots$ as would be obtained if we ran the decision algorithm on $\lambda=\lambda^*$.

Recall that for any $\lambda$, step $i$ of our decision algorithm
computes the sensor $s_{g(i)}$, the set
$S_i=\{s_{g(1)},s_{g(2)},\ldots,s_{g(i)}\}$, and the value $R_i$, and obtains the configuration $C_i$.
In the following,
we often consider $\lambda$ as a variable rather than a fixed value.
Thus, we will use $S_i(\lambda)$ (resp., $R_i(\lambda)$, $s_{g(i)}(\lambda)$,
$C_i(\lambda)$, $x_i^r(\lambda)$) to refer to the corresponding $S_i$
(resp., $R_i$, $s_{g(i)}$, $C_i$, $x_i^r$).
Our algorithm has at most $n$ steps. Consider a general $i$-th
step for $i\geq 1$. Right before the step, we have an interval
$(\lambda^1_{i-1},\lambda_{i-1}^2]$ and a sensor set $S_{i-1}(\lambda)$,
such that the following algorithm invariants hold.

\begin{enumerate}
\item
$\lambda^*\in (\lambda^1_{i-1},\lambda_{i-1}^2]$.

\item
The set $S_{i-1}(\lambda)$ is the same (with the same order) for all values
$\lambda\in (\lambda^1_{i-1},\lambda_{i-1}^2)$.
\item
$R_{i-1}(\lambda)$ on
$\lambda\in (\lambda^1_{i-1},\lambda_{i-1}^2)$ is either constant or equal to $x_j+ \sqrt{\lambda^2-y_j^2}+c$ for some constant $c$ and some sensor $s_j$, and the function $R_{i-1}(\lambda)$ is explicitly maintained by the algorithm.

\item
$R_{i-1}(\lambda)<\beta$ for all $\lambda\in (\lambda^1_{i-1},\lambda_{i-1}^2)$.
\end{enumerate}

Initially when $i=1$, we let $\lambda_0^1=\lambda^*_1$ and
$\lambda_0^2=\lambda^*_2$. Since $S_{0}(\lambda)=\emptyset$ and $R_0(\lambda)=0$
for any $\lambda$, by Lemma~\ref{lem:pre}, all invariants hold for $i=1$.
In general, the $i$-th step will either compute $\lambda^*$, or obtain an interval
$(\lambda^1_{i},\lambda_{i}^2]\subseteq (\lambda^1_{i-1},\lambda_{i-1}^2]$ and a
sensor $s_{g(i)}(\lambda)$ with $S_i(\lambda)=S_{i-1}(\lambda)\cup
\{s_{g(i)}(\lambda)\}$. The running time of the step is
$O((m+ n\log\log n)(\log n+\log m))$. The details are given below.

\subsection{The Algorithm}

We assume $\lambda^*\neq \lambda_{i-1}^2$ and thus $\lambda^*$
is in $(\lambda_{i-1}^1,\lambda_{i-1}^2)$. Our following
algorithm can proceed without this assumption and we make the assumption
only for explaining the rationale of our approach.
Since $\lambda^*\in (\lambda_{i-1}^1,\lambda_{i-1}^2)$, according to our
algorithm invariants, for all $\lambda\in (\lambda_{i-1}^1,\lambda_{i-1}^2)$, $S_{i-1}(\lambda)$ is the same as $S_{i-1}(\lambda^*)$.
We simulate the decision algorithm on $\lambda=\lambda^*$. To determine the sensor
$s_{g(i)}(\lambda^*)$, we first compute the set $S_{i1}(\lambda^*)$, as follows.

Consider any sensor $s_k$ in $S\setminus S_{i-1}(\lambda)$. Its position
in $C_{i-1}(\lambda)$ is $x_k^r(\lambda)=x_k+ \sqrt{\lambda^2-y_k^2}$,
which is an increasing function of $\lambda$.
Thus, both the left and the right extensions of $s_k$ in $C_{i-1}(\lambda)$
are increasing functions of $\lambda$.
Suppose $f(\lambda)$ is either the left or the right extension of $s_k$ in
$C_{i-1}(\lambda)$.
According to our algorithm invariants, $R_{i-1}(\lambda)$ on
$\lambda\in (\lambda^1_{i-1},\lambda_{i-1}^2)$ is either constant or equal to $x_j+ \sqrt{\lambda^2-y_j^2}+c$ for some constant $c$ and some sensor $s_j$. We claim that there is at most one value $\lambda$ in $(\lambda^1_{i-1},\lambda_{i-1}^2)$ such that
$R_{i-1}(\lambda)=f(\lambda)$. Indeed, if $R_{i-1}(\lambda)$ is constant, then this is obviously true; otherwise, this is also true because each of $f(\lambda)$ and $R_{i-1}(\lambda)$ on $\lambda\in [0,\infty)$ defines a half branch of a hyperbola (and thus they have at most one intersection in $(\lambda^1_{i-1},\lambda_{i-1}^2)$).

Let $S'=S\setminus S_{i-1}(\lambda)$.
If we increase $\lambda$ from $\lambda^1_{i-1}$ to
$\lambda^2_{i-1}$, an ``event'' happens if $R_{i-1}(\lambda)$ is equal to the
left or right extension value of a sensor $s_k\in S'$ at some value of
$\lambda$ (called an {\em event value}). Note that
$S_{i1}(\lambda)$ does not change between any two adjacent events.
To compute $S_{i1}(\lambda^*)$, we first compute all event values, and this can be done in $O(n)$ time by using the function $R_{i-1}(\lambda)$ and all left and right extension functions of the sensors in $S'$.
Let $\Lambda$ denote the set of all event values, and we also add $\lambda^1_{i-1}$ and $\lambda^2_{i-1}$ to $\Lambda$. We then
sort all values in $\Lambda$. Using the feasibility test in Lemma~\ref{lem:faster}, we do binary search  to find two adjacent values $\lambda_1$ and $\lambda_2$ in the sorted
list of $\Lambda$ such that $\lambda^*\in (\lambda_1,\lambda_2]$. Note that
$(\lambda_1,\lambda_2]\subseteq (\lambda^1_{i-1},\lambda_{i-1}^2]$.
Since $|\Lambda|=O(n)$, the binary search uses $O(\log n)$ feasibility tests, which takes overall $O(m\log n+ n\log n\log\log n)$ time.

We make another assumption that $\lambda^*\neq \lambda_2$. Again, this assumption is
only for the explanation and the following algorithm can proceed without this assumption. Under the assumption, for any $\lambda\in (\lambda_1,\lambda_2)$, the set $S_{i1}(\lambda)$ is exactly
$S_{i1}(\lambda^*)$. Hence, we can compute $S_{i1}(\lambda^*)$ by
taking any $\lambda\in (\lambda_1,\lambda_2)$ and explicitly
computing $S_{i1}(\lambda)$ in $O(n)$ time.

The above has computed  $S_{i1}(\lambda^*)$.
If $S_{i1}(\lambda^*)\neq \emptyset$, we take an arbitrary sensor of
$S_{i1}(\lambda^*)$ as $s_{g(i)}(\lambda^*)$. Further, we
let $\lambda_i^1=\lambda_1$, $\lambda_i^2=\lambda_2$, and
$S_i(\lambda)=S_{i-1}(\lambda)\cup \{s_{g(i)}(\lambda^*)\}$.

If $S_{i1}(\lambda^*)=\emptyset$, then we
need to compute the set $S_{i2}(\lambda^*)$.
Since $\lambda^*\in (\lambda_1,\lambda_2) \subseteq
(\lambda^1_{i-1},\lambda_{i-1}^2)$, according to our algorithm invariants,
$R_{i-1}(\lambda)$ is a nondecreasing function on $\lambda\in
(\lambda_1,\lambda_2)$. For each sensor
$s_k\in S$, $x_k- \sqrt{\lambda^2-y_k^2}-r$ is a decreasing function on $\lambda\in (\lambda_1,\lambda_2)$.
Therefore, the interval
$(\lambda_1,\lambda_2)$ contains at most one value $\lambda$ such that
$R_{i-1}(\lambda)=x_k- \sqrt{\lambda^2-y_k^2}-r$.
If we increase $\lambda$ from $\lambda_1$ to $\lambda_2$, an ``event'' happens when $R_{i-1}(\lambda)$
is equal to $x_k- \sqrt{\lambda^2-y_k^2}-r$ for some sensor $s_k\in S'$ at some {\em event
value} $\lambda$. Also, the set $S_{i2}(\lambda)$
is fixed between any two adjacent
events. Hence, we use the following way to compute $S_{i2}(\lambda^*)$.

We first compute the set $\Lambda$ of all event values, and
also add $\lambda_1$ and
$\lambda_2$ to $\Lambda$. After sorting all values of $\Lambda$, by using our decision algorithm, we do binary search  to find two adjacent values
$\lambda_1'$ and $\lambda_2'$ in the sorted list of $\Lambda$ with $\lambda^*\in (\lambda_1',\lambda_2']$. Note that $(\lambda_1',\lambda_2']\subseteq (\lambda_1,\lambda_2]$.
Since $|\Lambda|=O(n)$, the binary search calls the decision algorithm
$O(\log n)$ times, which takes $O(m\log n +n\log n\log\log n)$ time in total.
Since $S_{i2}(\lambda)$ is the same for all $\lambda \in (\lambda_1',\lambda_2')$, we take an
arbitrary value $\lambda\in (\lambda'_1,\lambda'_2)$ and
compute $S_{i2}(\lambda)$ explicitly in $O(n)$ time.

\begin{lemma}\label{lem:20}
If $S_{i2}(\lambda)= \emptyset$, then $\lambda^*$ is in
$\{\lambda_{i-1}^2, \lambda_2, \lambda_2'\}$.
\end{lemma}
\begin{proof}
Assume to the contrary that $\lambda^*\not\in \{\lambda_{i-1}^2, \lambda_2, \lambda_2'\}$. Then, our
previous two assumptions on $\lambda^*$ are true and $\lambda^*\in (\lambda_1',\lambda_2')$. According to our algorithm, $S_{i2}(\lambda^*)=S_{i2}(\lambda)=\emptyset$ (and $S_{i1}(\lambda^*)=S_{i1}(\lambda)=\emptyset$).
This means that if we applied the decision algorithm on $\lambda=\lambda^*$, the sensor $s_{g(i)}(\lambda^*)$ would not exist. In other words, the decision algorithm would stop after the first $i-1$ steps, i.e., the decision algorithm would only use sensors in $S_{i-1}(\lambda^*)$ to cover all barriers.

On the other hand, according to our algorithm invariants, $R_{i-1}(\lambda)<\beta$ for all $\lambda\in (\lambda_{i-1}^1,\lambda_{i-1}^2)$. Since  $\lambda^*\in (\lambda_1',\lambda_2')\subseteq (\lambda_{i-1}^1,\lambda_{i-1}^2)$, $R_{i-1}(\lambda^*)< \beta$, but this contradicts with that all barriers are covered by the sensors of $S_{i-1}(\lambda^*)$ after the first $i-1$ steps of the decision algorithm.
\qed
\end{proof}

By Lemma \ref{lem:20}, if $S_{i2}(\lambda)=\emptyset$, then $\lambda^*$ is the smallest
feasible value of $\{\lambda_{i-1}^2, \lambda_2, \lambda_2'\}$, which can be
found by performing three feasibility tests.  Otherwise, we proceed as follows.

We make the third assumption that $\lambda^*\neq \lambda_2'$.
Thus, $\lambda^*\in (\lambda_1',\lambda_2')$ and $S_{i2}(\lambda^*)=S_{i2}(\lambda)$ for any $\lambda \in (\lambda_1',\lambda_2')$.
Next, we compute $s_{g(i)}(\lambda^*)$, i.e., the leftmost sensor of $S_{i2}(\lambda^*)$. Although $S_{i2}(\lambda)$ is the same for all $\lambda\in (\lambda_1',\lambda_2')$, the leftmost sensor of $S_{i2}(\lambda)$ may not be the same for all $\lambda\in
(\lambda_1',\lambda_2')$.
For each sensor $s_k\in S_{i2}(\lambda)$ and any $\lambda\in
(\lambda_1',\lambda_2')$, the location of $s_k$ in the configuration
$C_{i-1}(\lambda)$ is $x_k^r(\lambda)$. As discussed before, $x_k^r(\lambda)$ for $\lambda\in (\lambda_1',\lambda_2')$ defines a piece of the upper branch of a hyperbola in the 2D
coordinate system in which the $x$-coordinates correspond to the
$\lambda$ values and the $y$-coordinates correspond to $x_k^r(\lambda)$ values.
 We consider the lower envelope $\calL$ of the functions $x_k^r(\lambda)$ defined by all sensors $s_k$ of $S_{i2}(\lambda)$. For each point $q$ of $\calL$, suppose $q$ lies on the function defined by a sensor $s_k$ and $q$'s $x$-coordinate is $\lambda_q$. If
$\lambda=\lambda_q$, then the leftmost sensor of $S_{i2}(\lambda)$ is
$s_k$. This means that each curve segment of $\calL$ defined by one sensor corresponds to the same leftmost sensor of $S_{i2}(\lambda)$. Based on this observation,
we compute $s_{g(i)}(\lambda^*)$ as follows.

Since the functions $x_k^r(\lambda)$ and $x_{j}^r(\lambda)$ of two sensors $s_k$ and $s_j$ have at most one intersection in $(\lambda_1',\lambda_2')$, the number of vertices of the lower envelope $\calL$ is $O(n)$ and $\calL$ can be computed in $O(n\log n)$ time  \cite{ref:AtallahSo85,ref:HershbergerFi89,ref:SharirDa96}.  Let $\Lambda$ be the set of the
$x$-coordinates of the vertices of $\calL$.
We also add $\lambda_1'$ and $\lambda_2'$ to $\Lambda$. After sorting all
values of $\Lambda$, by using our decision algorithm, we do binary search
on the sorted list of $\Lambda$ to find two adjacent values $\lambda_1''$ and $\lambda_2''$ such that $\lambda^*\in (\lambda_1'',\lambda_2'']$.
Note that $(\lambda_1'',\lambda_2'']\subseteq (\lambda_1',\lambda_2']$.
Since $\lambda_1''$ and $\lambda_2''$ are two adjacent values of the
sorted $\Lambda$,
by our above analysis, there is a sensor that is always the
leftmost sensor of $S_{i2}(\lambda)$ for all $\lambda\in
(\lambda_1'',\lambda_2'']$. To find the sensor, we can take any value $\lambda$ in $(\lambda_1'',\lambda_2'')$ and
explicitly compute the locations of sensors in $S_{i2}(\lambda)$.
The above computes $s_{g(i)}(\lambda^*)$ in $O(m\log n+ n\log n\log\log
n)$ time.

Finally, we let $\lambda_i^1=\lambda_1''$, $\lambda_i^2=\lambda_2''$, and
$S_i(\lambda)=S_{i-1}(\lambda)\cup \{s_{g(i)}(\lambda^*)\}$.

If the above computes $\lambda^*$, then we terminate the algorithm. Otherwise, we  obtain an interval $(\lambda_i^1,\lambda_i^2]\subseteq (\lambda_{i-1}^1,\lambda_{i-1}^2]$ that
contains $\lambda^*$ and the set $S_i(\lambda)$. If $s_{g(i)}(\lambda)\in S_{i1}(\lambda)$, then $R_i(\lambda)$ is equal to $x_{g(i)}+\sqrt{\lambda^2-y_{g(i)}^2}+r$.
If $s_{g(i)}(\lambda)\in S_{i2}(\lambda)$, then
$R_i(\lambda)=R_{i-1}(\lambda)+2r$. According to the third algorithm invariant, $R_{i-1}(\lambda)$ is either constant or equal to $x_j+ \sqrt{\lambda^2-y_j^2}+c$ for some constant $c$ and some sensor $s_j$. Therefore, regardless of $s_{g(i)}(\lambda)\in S_{i1}(\lambda)$ or $s_{g(i)}(\lambda)\in S_{i2}(\lambda)$, $R_{i}(\lambda)$ is either constant or equal to $x_{j'}+ \sqrt{\lambda^2-y_{j'}^2}+c'$ for some constant $c'$ and some sensor $s_{j'}$, which establishes the third algorithm invariant for $R_{i}(\lambda)$.

If it is not true that $R_i(\lambda)<\beta$ for all $\lambda\in (\lambda_i^1,\lambda_i^2)$, then we preform some additional processing as follows.
We first have the following lemma.

\begin{lemma}\label{lem:60}
If it is not true that
$R_i(\lambda)<\beta$ for all $\lambda\in (\lambda_i^1,\lambda_i^2)$,
then $R_i(\lambda)$ is strictly increasing on
$(\lambda_i^1,\lambda_i^2)$ and there is a single value $\lambda'\in
(\lambda_i^1,\lambda_i^2)$ such that $R_i(\lambda')=\beta$.
\end{lemma}
\begin{proof}
The proof is almost the same as that of Lemma 4 in \cite{ref:LeeMi17} and we include it here for the sake of completeness.

Since it is not true that
$R_i(\lambda)<\beta$ for all $\lambda\in (\lambda_i^1,\lambda_i^2)$, either
$R_i(\lambda)>\beta$ for all $\lambda\in (\lambda_i^1,\lambda_i^2)$, or there is  a value $\lambda'\in (\lambda_i^1,\lambda_i^2)$ with
$R_i(\lambda')=\beta$. We first argue that the former case cannot happen.

Assume to the contrary that $R_i(\lambda)>\beta$ for all
$\lambda\in (\lambda_i^1,\lambda_i^2)$. Then, $R_i(\lambda')>\beta$ for any $\lambda'\in (\lambda_i^1,\lambda^*)$ since $\lambda^*\in
(\lambda_i^1,\lambda_i^2]$. But this would
imply that we have found a feasible solution using only sensors in
$S_i(\lambda)$ and the maximum movement of all sensors in
$S_i(\lambda)$ is at most $\lambda'<\lambda^*$,
contradicting with that $\lambda^*$ is the maximum moving distance in an
optimal solution.

Hence, there is a value $\lambda'\in (\lambda_i^1,\lambda_i^2)$ with
$R_i(\lambda')=\beta$. Next, we show that $R_i(\lambda)$ must be a strictly
increasing function. Assume to the contrary this is not true. Then, $R_i(\lambda)$ must be constant on $(\lambda_i^1,\lambda_i^2)$. Thus, $R_i(\lambda)=\beta$ for
all $\lambda\in (\lambda_i^1,\lambda_i^2)$. Since $\lambda^*\in
(\lambda_i^1,\lambda_i^2]$, let $\lambda'$ be any value in
$(\lambda_i^1,\lambda^*)$. Hence, $R_i(\lambda')=\beta$, and as above, $\lambda'$ is a feasible value. However, $\lambda'<\lambda^*$ incurs
contradiction.
\qed
\end{proof}

By Lemma~\ref{lem:60}, we compute the value $\lambda'\in
(\lambda_i^1,\lambda_i^2)$ such that $R_i(\lambda')=\beta$.
This means that all barriers are covered by the sensors of $S_i(\lambda')$ in  $C_i(\lambda')$, and thus $\lambda'$ is a feasible value and $\lambda^*\in (\lambda_i^1,\lambda']$. Because $R_i(\lambda)$ is strictly
increasing, $R_i(\lambda)<\beta$ for all $\lambda\in
(\lambda_i^1,\lambda')$. We update $\lambda_i^2$ to $\lambda'$.

In either case, $R_i(\lambda)<\beta$ now holds for all $\lambda\in (\lambda_i^1,\lambda_i^2)$. Finally, we perform the jump-over procedure, as follows (this step is not needed in~\cite{ref:LeeMi17}).

If $R_i(\lambda)$ is a constant and $R_i(\lambda)$ is not in the interior of a barrier, then we set $R_i(\lambda)$ to the left endpoint of the next barrier. If $R_i(\lambda)$ is an increasing function, then we do the following.
If we increase $\lambda$ from $\lambda_i^1$ to $\lambda_i^2$, an {\em event} happens if $R_i(\lambda)$ is equal to the left or right endpoint of a barrier at some {\em event value} of $\lambda$. During the increasing of $\lambda$, between any two adjacent events, $R_i(\lambda)$ is either always in the interior of a barrier or is always between two barriers. We compute all event values in $O(m)$ time by using the function $R_{i}(\lambda)$ and the endpoints of all barriers.
Let $\Lambda$ denote the set of all event values, and we also add $\lambda^1_{i}$ and $\lambda^2_{i}$ to $\Lambda$. After sorting all values in $\Lambda$, using the decision algorithm in Lemma~\ref{lem:faster}, we do binary search on the sorted list of $\Lambda$ to
find two adjacent values $\lambda_1$ and $\lambda_2$  such that $\lambda^*\in (\lambda_1,\lambda_2]$. Note that
$(\lambda_1,\lambda_2]\subseteq (\lambda^1_{i},\lambda_{i}^2]$.
Since $|\Lambda|=O(m)$, the binary search calls the decision algorithm $O(\log m)$ times, which takes overall $O(m\log m+ n\log\log n\log m)$ time.
Finally, we reset $\lambda_i^1=\lambda_1$ and $\lambda_i^2=\lambda_2$.

This completes the $i$-th step of the algorithm, which runs in $O((m+n\log\log
n)\cdot (\log m+\log n))$ time. If $\lambda^*$ is not computed in this step, then it can be verified that all algorithm variants are maintained (the analysis is similar to that in~\cite{ref:LeeMi17}, so we omit it). The algorithm will compute $\lambda^*$ after at most $n$ steps as there are $n$ sensors. The total time of the algorithm is $O(n\cdot (m+n\log\log n)\cdot (\log m+\log n))$, which is bounded by $O(nm\log m + n^2\log n\log\log n)$ as shown in the following theorem. Note that the space of the algorithm is $O(n)$.

\begin{theorem}
The problem MBC can be solved in $O(nm\log m + n^2\log n\log\log n)$ time and $O(n)$ space.
\end{theorem}
\begin{proof}
As discussed before, the running time of the algorithm is $O(n\cdot (m+n\log\log n)\cdot (\log m+\log n))$, which is $O(nm\log m+n^2\log n\log\log n+ nm\log n+n^2 \log m\log \log n)$. We claim that $nm\log n+n^2 \log m\log \log n=O(nm\log m+n^2\log n\log\log n)$.
Indeed, if $m\leq n\log\log n$, then $nm\log n=O(n^2 \log n\log \log n)$ and $n^2 \log m\log \log n=O(n^2 \log n\log \log n)$; otherwise, $nm\log n=O(nm\log m)$ and $n^2 \log m\log \log n=O(nm\log m)$.
\qed
\end{proof}

\section{Concluding Remarks}
\label{sec:conclusion}

As mentioned before, the high-level scheme of our algorithm for MBC is similar to those in~\cite{ref:LeeMi17}. However, a new technique we propose in this paper can help reduce the space complexity of the algorithm in~\cite{ref:LeeMi17}. Specifically, Lee et al.~\cite{ref:LeeMi17} solved the line-constrained problem in $O(n^2\log n\log\log n)$ time and $O(n^2)$ space for the case where $m=1$, sensors have the same range, and sensors have weights. If we apply the similar preprocessing as in Lemma~\ref{lem:pre}, then the space complexity of the algorithm~\cite{ref:LeeMi17} can be reduced to $O(n)$ while the time complexity does not change asymptotically.

%

\bibliographystyle{plain}
\bibliography{reference-TCS}

\end{document}